\documentclass{llncs}
\usepackage{color}
\usepackage{amsmath}
\usepackage{amssymb}
\usepackage{amsfonts}
\usepackage{centernot}
\usepackage[ruled,vlined]{algorithm2e}
\usepackage{algpseudocode}
\usepackage[arxiv]{optional}

\usepackage{color}
\usepackage[font={small},labelfont={bf}, labelsep=period]{caption}
\usepackage[font={small},labelfont={small}, labelformat=simple,subrefformat=parens,labelformat=parens]{subfig}
\usepackage{wrapfig}
\usepackage{todonotes}
\usepackage{url}

\usepackage{enumerate}
\usepackage{xspace}
\usepackage{rotating}

\usepackage{todonotes}
\spnewtheorem*{sketchofproof}{Sketch of Proof}{\itshape}{\rmfamily}
\newcommand{\rephrase}[3]{\medskip\noindent\textbf{#1~#2.}~\emph{#3}}

\usepackage{graphicx}
\graphicspath{{./fig/}}

\newcommand{\gforall}{\forall\:}

\newcommand{\pathto}{\rightsquigarrow}
\newcommand{\npathto}{\centernot\rightsquigarrow}

\begin{document}

\opt{arxiv}{\pagestyle{plain}}

%%%% Title
\title{Bitonic \textit{st}-orderings for Upward Planar Graphs}

%%%% Author list of contribution
\author{Martin Gronemann}
%
%%%% Author list for running head
\authorrunning{M. Gronemann}

\institute{University of Cologne, Germany\\
\email{gronemann@informatik.uni-koeln.de}} 

\maketitle

\begin{abstract}
Canonical orderings serve as the basis for many incremental planar drawing algorithms.
All these techniques, however, have in common that they are limited to undirected graphs.
While $st$-orderings do extend to directed graphs, especially planar $st$-graphs, they do not
offer the same properties as canonical orderings. In this work we extend the so called 
bitonic $st$-orderings to directed graphs. We fully characterize planar $st$-graphs that 
admit such an ordering and provide a linear-time algorithm for recognition and ordering.
If for a graph no bitonic $st$-ordering exists, we show how to find in linear time a 
minimum set of edges to split such that the resulting graph admits one.
With this new technique we are able to draw every upward planar graph on $n$ vertices by
using at most one bend per edge, at most $n - 3$ bends in total and within quadratic area.
\end{abstract}

\section{Introduction}
Drawing directed graphs is a fundamental problem in graph drawing and 
has therefore received a considerable amount of attention in the past.
Especially the so called \emph{upward planar drawings}, a planar drawing in which the 
curve representing an edge has to be strictly $y$-monotone 
from its source to target. The directed graphs that admit such a drawing
are called the \emph{upward planar} graphs. Deciding if a directed graph is upward planar
turned out to be {NP}-complete in the general case~\cite{Garg:1995fk}, but there exist special cases
for which the problem is polynomial-time solvable~\cite{Abbasi2010274,Bertolazzi:1994uq,DidimoGL09,hl-uptssad-96,p-uptod-95,DBLP:conf/walcom/SameeR07}. An important result in our context is from Di Battista and Tamassia~\cite{DiBattista1988175}. They show that every upward planar
graph is the spanning subgraph of a planar $st$-graph, that is, a planar directed acyclic graph with a single source
and a single sink. They also show that every such graph has an 
upward planar straight-line drawing~\cite{DiBattista1988175}, but it may require exponential area
which for some graphs cannot be avoided~\cite{2014arXiv1410.1006D,Battista:1992fk}.

If one allows bends on the edges, then every upward planar graph can be drawn
within quadratic area.
Di~Battista and Tamassia~\cite{DiBattista1988175} describe an approach that is based 
on the visibility representation of a planar $st$-graph. Every
edge has at most two bends, therefore, the resulting drawing has at most $6n-12$ bends with $n$ being the number of vertices.
With a more careful choice of the vertex positions and by employing a special visibility representation,
the authors manage to improve this bound to $(10n - 31)/3$. Moreover, the drawing requires only quadratic area and can be obtained in linear time.
Another approach by Di~Battista~et~al.~\cite{Battista:1992fk} uses an algorithm 
that creates a straight-line dominance drawing as an intermediate step. 
A dominance drawing, however, has much stronger requirements than an upward planar drawing. 
Therefore, the presented algorithm in~\cite{Battista:1992fk} cannot handle planar $st$-graphs directly.
Instead it requires a \emph{reduced planar $st$-graph}, that is, a planar $st$-graph without \emph{transitive edges}. 
In order to obtain such a graph, Di~Battista~et~al.~\cite{Battista:1992fk} split every transitive edge
by replacing it with a path of length two. The result is a reduced planar $st$-graph 
for which a straight-line dominance drawing is obtained that requires only quadratic area 
and can be computed in linear time. Then they reverse the procedure of splitting the edges 
by using the coordinates of the inserted dummy vertices as bend points. 
Since a planar $st$-graph has at most $2n-5$ transitive edges, 
the resulting layout has not more than $2n-5$ bends and at most
one bend per edge. To our knowledge, this bound is the best achieved so far.

These techniques are very different to the ones used in the undirected case. 
One major reason is the availability of \emph{canonical orderings} for undirected graphs,
introduced by de Fraysseix~et~al.~\cite{fpp-hdpgg-90} to draw every (maximal) 
planar graph straight-line within quadratic area. From there on this concept has been further 
improved and generalized~\cite{harel98algorithm,k-dpguc-96,KANT1997175}. Biedl and Derka~\cite{BiedlD15a} discuss various variants
and their relation. Another similar concept that extends to non-planar graphs is the Mondshein sequence~\cite{Schmidt2014}.
However, all these orderings have in common that they do not extend to directed graphs,
that is, for every edge $(u,v)$, it holds that $u$ precedes $v$ in the ordering. An exception
are $st$-orderings. While they are easy to compute for planar $st$-graphs, they lack 
a certain property compared to canonical orderings. In~\cite{gronemann14} we introduced
for undirected biconnected planar graphs the \emph{bitonic $st$-ordering}, a special $st$-ordering which has 
properties similar to canonical orderings. However, the algorithm in~\cite{gronemann14}
uses canonical orderings for the triconnected case as a subroutine. Since finding a canonical ordering
is in general not a trivial task, respecting the orientation of edges makes it even harder.
Nevertheless, such an ordering is desirable, since one would be able to use incremental drawing approaches
for directed graphs that are usually limited to the undirected case.

In this paper we extend the bitonic $st$-ordering to directed graphs, namely planar $st$-graphs. 
We start by discussing the consequences of having such an ordering available.
Based on the observation that the algorithm of de~Fraysseix et al.~\cite{fpp-hdpgg-90} can easily be modified to obtain an
upward planar straight-line drawing, we show that for good reasons 
not every planar $st$-graph admits such an ordering.
After deriving a full characterization of the planar $st$-graphs that do admit a bitonic $st$-ordering,
we provide a linear-time algorithm that recognizes these and computes a corresponding ordering.
For a planar $st$-graph that does not admit a bitonic $st$-ordering, we show that splitting at most $n-3$
edges is sufficient to transform it into one for which then an ordering can be found.
Furthermore, a linear-time algorithm is described that determines the smallest set of edges to split.
By combining these results, we are able to draw every planar $st$-graph with at most one bend per edge, $n-3$ bends in total
within quadratic area in linear time. This improves the upper bound on the total number of bends considerably.
%================= lncs version ==========
%Some proofs have been omitted and can be found in the full version~\cite{bitonic-upward-arxiv} or in~\cite{Gronemann15}.
\opt{lncs}{Some proofs have been omitted and can be found in the full version~\cite{bitonic-upward-arxiv} or in~\cite{Gronemann15}.}
%================= arxiv version ==========
\opt{arxiv}{Some proofs have been omitted and can be found in Appendix~\ref{app:proofs} or in~\cite{Gronemann15}.}
%======================================

\section{Preliminaries}
In this work we are solely concerned with a special type of directed graph, the so-called \emph{planar st-graph}, 
that is, a planar acyclic directed graph $G= (V,E)$ with a single source $s \in V$, a single sink $t \in V$ and no parallel edges.
It should be noted that some definitions assume that $(s,t) \in E$, 
we explicitly do not require this edge to be present. However, we assume a fixed embedding scenario
such that $s$ and $t$ are on the outer face. Under such constraints, planar $st$-graphs possess 
the property of being \emph{bimodal}, that is, the incoming and outgoing edges
appear as a consecutive sequence around a vertex in the embedding.
Given an edge $(u,v) \in E$, we refer to $v$ as a \emph{successor} of $u$ and call $u$ a \emph{predecessor} of $v$.
Similar to~\cite{gronemann14}, we define for every vertex $u \in V$ a list of successors $S(u) = \{ v_1, \ldots, v_m\}$,
ordered by the outgoing edges $(u,v_1), \ldots, (u,v_m)$ of $u$ as they appear in the embedding clockwise around $u$.
For $S(s)$ we choose $v_1$ and $v_m$ such that $v_m, s, v_1$ appear clockwise on the outer face.
A central problem will be the existence of paths between vertices.
Therefore, we refer to a path from $u$ to $v$ and its existence with $u \pathto v \in G$.
With a few exceptions, $G$ is clear from the context, thus, we omit it. If there exists no path
$u \pathto v$, we may abbreviate it by writing $u \npathto v$.

%====================================================================
%      The bitonic ST-ordering
%====================================================================
Let $G = (V,E)$ be a planar $st$-graph and $\pi : V \mapsto \{1, \ldots, |V| \}$ 
be the rank of the vertices in an ordering $s = v_1, \ldots, v_n = t$. 
$\pi$ is said to be an \emph{$st$-ordering}, if for all edges $(u,v) \in E$, $\pi(u) < \pi(v)$ holds.
In case of a (planar) $st$-graph such an ordering can be obtained in linear time 
by using a simple topological sorting algorithm~\cite{Cormen:2009}. We are interested in a special type
of $st$-ordering, the so called bitonic $st$-ordering introduced in~\cite{gronemann14}. We say an ordered sequence 
$A = \{ a_1, \ldots, a_n \}$ is \emph{bitonic increasing},
if there exists $1 \leq h \leq n$ such that $a_1 \leq \cdots \leq a_h \geq \cdots \geq a_n$
and \emph{bitonic decreasing}, if $a_1 \geq \cdots \geq a_h \leq \cdots \leq a_n$.
Moreover, we say $A$ is bitonic increasing (decreasing) with respect to a function 
$f$, if $A' = \{ f(a_1), \ldots, f(a_n) \}$ is bitonic increasing (decreasing).
In the following, we restrict ourselves to bitonic increasing sequences and abbreviate it by just referring to it as being bitonic.
An $st$-ordering $\pi$ for $G$ is a \emph{bitonic $st$-ordering} for $G$, 
if at every vertex $u \in V$ the ordered sequence of successors $S(u) = \{ v_1, \ldots, v_m \}$ 
as implied by the embedding is bitonic with respect to $\pi$, that is, there exists 
$1 \leq h \leq m$ with $\pi(v_1) < \cdots < \pi(v_h) > \cdots > \pi(v_m)$. 
Notice that the successors of a vertex are distinct and so are their labels in an $st$-ordering.

%===================================================================================
\section{Upward planar straight-line drawings \& bitonic \textit{st}-orderings}\label{sec:bitonic_straight_line}
%===================================================================================
We start by assuming that we are given a planar $st$-graph $G = (V,E)$ together with a bitonic $st$-ordering $\pi$.
The idea is to use the straight-line algorithm from~\cite{gronemann14} which is based on the one 
in~\cite{harel98algorithm} to produce an upward planar straight-line layout.
Due to space constraints, we omit details here and only sketch the two modifications that are
necessary. For a full pseudocode listing, an example and a detailed description, 
%=========== lncs =============
\opt{lncs}{see the full version~\cite{bitonic-upward-arxiv} or~\cite{Gronemann15}.}%
%=========== arxiv =============
\opt{arxiv}{see Appendix~\ref{app:desc} or~\cite{Gronemann15}.}
%=============================
When using a bitonic $st$-ordering to drive the planar straight-line algorithm of de~Fraysseix~et~al.~\cite{fpp-hdpgg-90},
the only critical case is the one in which a vertex $v_k$ must be placed that has only one neighbor, say $w_i$, in the subgraph drawn so far.
In~\cite{gronemann14} we use the idea of Harel and Sardas~\cite{harel98algorithm} who guarantee with their ordering that 
the edges preceding or following $(w_i, v_k)$ in the embedding around $w_i$ have already been drawn. 
Hence one may just pretend that $v_k$ has a second neighbor either to the right or left of $w_i$. 
The idea is illustrated in Fig.~\ref{fig:fpp_bitonic_support_1} where $v_k$ uses $w_{i+1}$, the successor of $w_i$ on the contour, as second neighbor.
The following lemma captures the required property and shows that a bitonic $st$-ordering complies with it.

\begin{figure}[t]
  \centering
    \begin{minipage}[b]{0.48\textwidth}
        \centering
        \subfloat[\label{fig:fpp_bitonic_support_1}{}]
        { \includegraphics[page=1]{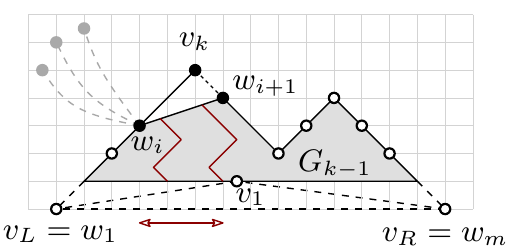}}
    \end{minipage}\hfill
    \begin{minipage}[b]{0.48\textwidth}
        \centering
        \subfloat[\label{fig:fpp_bitonic_support_2}{}]
        {\includegraphics[page=2]{fpp_property}}
    \end{minipage}
\caption{(a)~A vertex $v_k$ with only one predecessor $w_i$ using the vertex $w_{i+1}$ as second neighbor. 
		   Vertices in grey have not been drawn yet. The two dummy vertices $v_L, v_R$ remain the left- and rightmost ones.
	      (b)~Example of an upward planar straight-line drawing on seven vertices.}
	      \label{fig:fpp_bitonic_support}
\end{figure}

%=================================================================
\newcommand{\lembitonicdrawing}{
Let $G = (V,E)$ be an embedded planar $st$-graph with a corresponding bitonic $st$-ordering $\pi$.
Moreover, let $v_k$ be the {$k$-th} vertex in $\pi$ and $G_k=(V_k,E_k)$ the subgraph induced by $v_1, \ldots, v_k$.
For every $1 < k \leq |V|$ the following holds:
\begin{enumerate}
\item $G_{k}$ and $G-G_{k}$ are connected,
\item $v_k$ is in the outer face of $G_{k-1}$,
\item For every vertex $v \in V_k$, the neighbors of $v$ that are not in $G_k$
appear consecutively in the embedding around $v$.
\end{enumerate}}
%=================================================================

\begin{lemma}\label{lem:lembitonicdrawing}
\lembitonicdrawing
\end{lemma}

\begin{sketchofproof}
The first two properties hold for all $st$-orderings. For the third, assume to the contrary, contradicting that 
$S(v)$ is bitonic with respect to $\pi$.
\qed\end{sketchofproof}

Due to the third statement we can always choose a second neighbor either to the left or right, since otherwise the grey vertices in Fig.~\ref{fig:fpp_bitonic_support_1}
would not be consecutive in the embedding around $w_i$.
The second modification solves a problem that arises in the initialization phase
of the drawing algorithm. Recall that in~\cite{fpp-hdpgg-90} the first three vertices are drawn as a triangle.
This of course works in the case of a canonical ordering, but requires extra care when using a bitonic $st$-ordering.
In order to avoid subcases and keep things simple,
we add two isolated dummy vertices $v_L$ and $v_R$ that take the roles of the first two vertices 
and pretend to form a triangle with $v_1 = s$. This has another side effect: It avoids distinguishing
between subcases when we have to find a second neighbor at the boundary of the contour, because 
$v_L$ is always the first, and $v_R$ always the last vertex on every contour during the incremental construction. See the example in Fig.~\ref{fig:fpp_bitonic_support_2}.

\begin{theorem}\label{thm:bitonic_upward_straightline}
Given an embedded planar $st$-graph $G = (V,E)$ and a corresponding bitonic $st$-ordering $\pi$ for $G$.
An upward planar straight-line drawing for $G$ of size $(2|V|-2) \times (|V|-1)$
can be obtained from $\pi$ in linear time.
\end{theorem}
\begin{proof}
The upward property is obtained by the following observation: The original planar straight-line algorithm
installs every vertex $v_k$ with $k > 2$ above its predecessors. Since we start with $v_L, v_R, v_1$,
the drawing is upward. It remains to bound the area. 
Notice that the input consists of the two additional vertices $v_L,v_R$.
The original algorithm, without any area improvements, produces a drawing with a size of 
$2((|V|+2) - 4) \times (|V|+2) - 2$ = $2|V| \times |V|$. 
However, $v_L$ and $v_R$ are dummy vertices  and can be removed anyway. 
Moreover, every other vertex is located above them. Hence,
their removal yields a smaller drawing of size $(2|V|-2) \times (|V|-1)$.
\qed\end{proof}

Now the first question that comes to mind is, if we can always find a bitonic $st$-ordering.
Although every planar $st$-graph admits an upward planar straight-line drawing~\cite{DiBattista1988175},
there exist some classes for which it is known that they require exponential area~\cite{2014arXiv1410.1006D,Battista:1992fk}. 
Since Theorem~\ref{thm:bitonic_upward_straightline} clearly states that the drawing 
requires only polynomial area, these graphs cannot admit a bitonic $st$-ordering.

\begin{corollary}
Not every planar $st$-graph admits a bitonic $st$-ordering.
\end{corollary}

While this had to be expected, we now have to solve an additional problem. 
Before we think about how to compute a bitonic $st$-ordering,
we must first be able to recognize planar $st$-graphs that admit such an ordering.
%===================================================================================
\section{Characterization, recognition \& ordering}
%===================================================================================
We proceed as follows: As a first step, we identify a necessary condition
that a planar $st$-graph has to meet for admitting a bitonic $st$-ordering.
Then we exploit this condition to compute a bitonic $st$-ordering which proves sufficiency.
We start with an alternative characterization of bitonic sequences.
Since we will use the labels of an $st$-ordering, we can assume that the elements
are pairwise distinct. 
%=================================================================
\newcommand{\lembitonic}{
An ordered sequence $A = \{ a_1, \ldots, a_n \}$ of pairwise distinct elements is bitonic increasing
if and only if the following holds:
\[
\forall 1 \leq i < j < n \; : \; a_{i} < a_{i+1} \vee a_{j} > a_{j+1}.
\]}
%=================================================================
\begin{lemma}\label{lem:bitonic_alternative}
\lembitonic
\end{lemma}
\begin{sketchofproof}
For~\lq\lq$\Rightarrow$\rq\rq, assume to the contrary which yields $i \geq j$.
For~\lq\lq$\Leftarrow$\rq\rq, we choose, if exists, $h = \min\{j \; | \; a_j > a_{j+1}\}$, otherwise
we set $h = n$.\qed
\end{sketchofproof}

In general a planar $st$-graph may have many $st$-orderings,
some of them being bitonic while others are not.
To deal with this in a more formal manner,
we introduce some additional notation.
Given an embedded planar $st$-graph $G = (V,E)$,
we refer with $\Pi(G)$ to all feasible $st$-orderings of $G$, that is,
\[
\Pi(G) = \{ \pi : V \mapsto \{1, \ldots, |V|\} \;| \;\pi \text{ is an $st$-ordering for } G \}.
\]
Furthermore, let $\Pi_b(G)$ be the subset of $\Pi(G)$ that contains
all bitonic $st$-orderings. By definition, we can describe $\Pi_b(G)$ by
\begin{equation*}
\Pi_b(G) = \{ \pi \in \Pi(G) \mid \forall u \in V \; : \; S(u) \text{ is bitonic with respect to } \pi\}.
\end{equation*}

\noindent Applying the alternative characterization of bitonicity from Lemma~\ref{lem:bitonic_alternative}
to the bitonic property of the successor lists $S(u)$ yields the following expression for
the existence of a bitonic $st$-ordering:
\begin{equation}\label{eq:bitonic_st_1}
\begin{split}
\exists \pi \in \Pi_b(G) \Leftrightarrow&\;\exists \pi \in \Pi(G) \quad \forall u \in V \text{ with } S(u)=\{v_1, \ldots, v_m \}\\
 &\;\gforall 1\leq i < j < m \;:\; \pi(v_{i}) < \pi(v_{i+1}) \vee \pi(v_{j}) > \pi(v_{j+1}).
\end{split}
\end{equation}

\begin{figure}[t]
\centering
   \begin{minipage}[b]{.39\textwidth}
        \centering
        \subfloat[\label{fig:bitonic_forbidden_config}{}]
        {\includegraphics[page=1, scale=0.9]{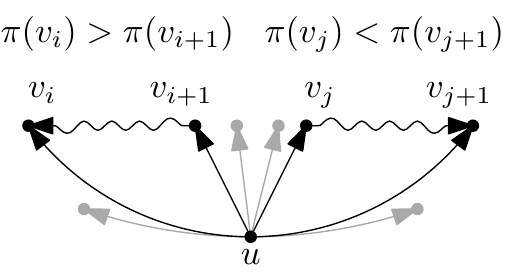}}
    \end{minipage}
    \begin{minipage}[b]{.19\textwidth}
        \centering
        \subfloat[\label{fig:bitonic_st_condition_1}{}]
        {\includegraphics[page=2, scale=0.9]{bitonic_st_condition_new}}
    \end{minipage}
    \begin{minipage}[b]{.19\textwidth}
        \centering
        \subfloat[\label{fig:bitonic_st_condition_2}{}]
        {\includegraphics[page=3, scale=0.9]{bitonic_st_condition_new}}
    \end{minipage}
     \begin{minipage}[b]{.19\textwidth}
        \centering
        \subfloat[\label{fig:bitonic_st_condition_3}{}]
        {\includegraphics[page=4, scale=0.9]{bitonic_st_condition_new}}
    \end{minipage}
    \caption{(a)~A successor list $S(u) = \{ \ldots, v_i, v_{i+1}, \ldots, v_j, v_{j+1}, \ldots \}$ with $i < j$ and a
			forbidden configuration of paths $v_{i+1} \rightsquigarrow v_{i}$ and $v_j \rightsquigarrow v_{j+1}$.
    		 (b)-(d)~The three cases at a face between two successors $v_i$ and $v_{i+1}$ of the face-source $u$:             
                 (b)~$v_{i+1}$ is the sink of the face indicating the existence of a path from $v_{i}$ to $v_{i+1}$.
                 (c)~A path from $v_{i+1}$ to $v_{i}$ results in a face having $v_i$ as sink.
                 (d)~There exists no path between $v_i$ and $v_{i+1}$, if and only if neither $v_i$ nor $v_{i+1}$ is the face-sink.}\label{fig:bitonic_st_condition_x}
\end{figure}

\noindent Next we translate this expression from $st$-orderings to the existence of paths. 
Consider a path from some vertex $u$ to some other vertex $v$ in $G$,
then for every $\pi \in \Pi(G)$, by the definition of $st$-orderings, $\pi(u) < \pi(v)$ holds.
Now it is not hard to imagine that if there exists $\pi \in \Pi_b(G)$, then
there must exist configurations of paths that are forbidden. 
To clarify this, let us rewrite the last part of the condition in Equation~\ref{eq:bitonic_st_1}, that is, $\pi(v_{i}) < \pi(v_{i+1}) \vee \pi(v_{j}) > \pi(v_{j+1})$,
using a simple boolean transformation, which yields \mbox{$\neg (\pi(v_{i}) > \pi(v_{i+1}) \wedge \pi(v_{j}) < \pi(v_{j+1}))$}.
So if there exists a path from $v_{i+1}$ to $v_i$ and one from $v_j$ to $v_{j+1}$ with $i < j$,
then this expression evaluates to false for every $\pi \in \Pi(G)$. 
Therefore, we may refer to the pair of paths $v_{i+1} \rightsquigarrow v_i$ and $v_j \rightsquigarrow v_{j+1}$ with $i < j$ 
as a \emph{forbidden configuration} of paths.
See Fig.~\ref{fig:bitonic_forbidden_config} for an illustration.

We may state now that in case there exists a bitonic $st$-ordering, the aforementioned configuration of paths cannot exist:
\begin{equation*}
\begin{split}
\exists \pi \in \Pi_b(G) \Rightarrow&\; \forall u \in V \text{ with } S(u)=\{v_1, \ldots, v_m \}\\
 &\;\gforall 1\leq i < j < m \;:\; v_{i+1} \npathto v_{i} \vee v_{j} \npathto v_{j+1}.
\end{split}
\end{equation*}
Conversely, if we find an $u$ with $v_i$ and $v_j$ in a graph for which these paths exist, 
then we can safely reject it as one that does not admit a bitonic $st$-ordering. 
The following well-known property of planar $st$-graphs will prove itself useful when it comes to 
testing for the existence of a path between two vertices.
\begin{lemma}\label{lem:bitonic_face_paths}
Let $F$ be the subgraph of an embedded planar $st$-graph $G = (V,E)$ 
induced by a face that is not the outer face\footnote{This restriction is necessary due to the possible absence of the $st$-edge
which is allowed by our definition of planar $st$-graphs.}, and $u,v$ two vertices of 
$F$, that is, $u$ and $v$ are on the boundary of the face. Then there exists
a path from $u$ to $v$ in $G$, if and only if there exists such a path in $F$.
\end{lemma}

There are several ways to prove this result, one proof can be found in the
work of de~Fraysseix et al.~\cite{deFraysseix1995157}. 
Notice that Lemma~\ref{lem:bitonic_face_paths} is concerned with every pair of vertices incident to the face.
But we are only interested in paths between two consecutive successors $v_{i}$ and $v_{i+1}$
of a vertex $u$. Notice that $v_{i}, v_{i+1}$ and $u$ share a common face which is not the 
outer face and in which $u$ is the face-source. Fig.~\ref{fig:bitonic_st_condition_1}-d illustrates
all three possible cases: $v_i \pathto v_{i+1}$~\subref{fig:bitonic_st_condition_1}, $v_{i+1} \pathto v_{i}$~\subref{fig:bitonic_st_condition_2}, 
and no path at all~\subref{fig:bitonic_st_condition_3}. Hence, we can decide the existence of a path based on the sink of the common face.

To prove that the absence of forbidden configurations is sufficient for the existence of a bitonic $st$-ordering, we require 
the following technical proposition.

%=================================================================
\newcommand{\probitonicpaths}{
Given an embedded planar $st$-graph $G = (V,E)$ and a vertex $u \in V$ with successor list $S(u) = \{v_1, \ldots, v_m \}$.
If it holds that
\begin{equation*}
\gforall 1\leq i < j < m \; : \; v_{i+1} \npathto v_{i} \vee v_{j} \npathto v_{j+1},
\end{equation*}
then there exists $1 \leq h \leq m$ such that
\begin{equation*}
	(\gforall 1 \leq i < h :  v_{i+1} \npathto v_{i}) \wedge (\gforall h \leq i < m :  v_{i} \npathto v_{i+1})
\end{equation*}
holds. In other words, there exists at least one $v_h$ in $S(u)$ whose preceding vertices in $S(u)$ are only connected by paths in
clockwise direction, whereas paths between following vertices are directed counterclockwise.}
%=================================================================

\begin{proposition}\label{pro:bitonic_h}
\probitonicpaths
\end{proposition}
\begin{sketchofproof}
If exists, set ${ h = \min\{ i \mid  v_{i+1} \pathto v_{i} \} }$, otherwise set $h = m$.
\qed\end{sketchofproof}

\begin{figure}[t]
\centering
    \begin{minipage}[b]{.48\textwidth}
        \centering
        \subfloat[\label{fig:bitonic_h_example}{}]
         {\includegraphics[page=1,scale=0.9]{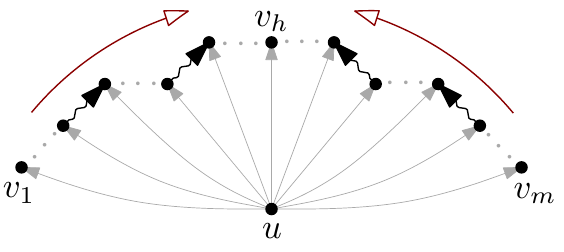}}
    \end{minipage}\hfill
    \begin{minipage}[b]{.48\textwidth}
        \centering
        \subfloat[\label{fig:make_bitonic_edges}{}]
         {\includegraphics[page=2,scale=0.9]{proposition_lemma}}
    \end{minipage}
    \caption{(a)~Paths orientations between consecutive successors of $u$. 
    			All of them directed towards $v_h$ as described by Proposition~\ref{pro:bitonic_h}.
                  (b)~The augmented graph $G'$ in the proof of Lemma~\ref{lem:bitonic_planar_st_graph_ordering} 
                  obtained by adding edges between consecutive successors of $u$ such that they are oriented towards $v_h$.}
                  \label{fig:bitonic_st_condition_0}
\end{figure}

The idea is now the following: If we have a graph that satisfies our necessary condition, 
then we can find for every $u \in V$ with $u \neq t$ a successor $v_h$ with the property 
as described in Proposition~\ref{pro:bitonic_h}. The intuition behind this property is that all paths 
that exist between successors of $u$, are directed in some way towards $v_h$. 
See Fig.~\ref{fig:bitonic_h_example} for an illustration.
The next lemma exploits this property to obtain a bitonic $st$-ordering,
which proves that this condition is indeed sufficient for the existence of a bitonic $st$-ordering.

\begin{lemma}\label{lem:bitonic_planar_st_graph_ordering}
Given a planar $st$-graph $G=(V,E)$ with a fixed embedding. 
If at every vertex $u \in V$ with successor list $S(u)=\{v_1, \ldots, v_m \}$ the following holds:
\[
\gforall 1\leq i < j < m \;:\;  v_{i+1} \npathto v_{i} \vee v_{j} \npathto v_{j+1},
\]
then $G$ admits a bitonic $st$-ordering $\pi \in \Pi_b(G)$.
\end{lemma}
\begin{proof}
To show that there exists $\pi \in \Pi_b(G)$, we augment $G$ into a new graph $G'$
by inserting additional edges that we refer to as $E'$.
These edges ensure that between every pair of consecutive successors in $G$, there exists a path in $G' = (V, E \cup E')$.
Afterwards, we show that every $st$-ordering $\pi \in \Pi(G')$ for $G'$ is a bitonic $st$-ordering for $G$.

For every vertex $u$ with successor list $S(u) = \{v_1, \ldots, v_m\}$, we may assume 
by Proposition~\ref{pro:bitonic_h} that there exists $1 \leq h \leq m$ 
such that for every $1 \leq i < h$ there exists no path from $v_{i+1}$ to $v_{i}$, and 
for every $h \leq i < m$ no path from $v_{i}$ to $v_{i+1}$ in $G$. 
Our goal is to add specific edges to fill the gaps such that there exist two paths in $G'$,
$v_1 \rightsquigarrow v_2 \rightsquigarrow \cdots \rightsquigarrow v_h \in G'$ and 
$v_m \rightsquigarrow v_{m-1} \rightsquigarrow \cdots \rightsquigarrow v_h \in G'$. 
Fig.~\ref{fig:make_bitonic_edges} illustrates the idea.
More specifically, for every $1 \leq i < m$, there are three cases to consider:
(i)~There already exists a path between $v_i$ and $v_{i+1}$ in $G$, that is,
$v_i \rightsquigarrow v_{i+1} \in G$ or $ v_{i+1} \rightsquigarrow v_{i} \in G$.
Proposition~\ref{pro:bitonic_h} ensures that the path is directed towards $v_h$, thus, we just skip the pair.
(ii)~If there exists no path between $v_i$ and $v_{i+1}$ in $G$ and $i < h$ holds, 
we add an edge from $v_i$ to $v_{i+1}$. (iii)~When there also exists no path between 
$v_i$ and $v_{i+1}$, but now $h \leq i < m$ holds, we add the reverse edge $(v_{i+1}, v_i)$ to $E'$.

Before we continue, we show that $G' = (V, E \cup E')$ is $st$-planar.
Consider a single edge in $E'$ which has been added either by case (ii) or (iii) 
while traversing the successors $S(u)$ of some vertex $u \in V$. 
This edge will be added to a face in which $u$ is the source, and since every face
has only one source, only one edge will be added to the corresponding face, hence, planarity is preserved.
Since case (ii) and (iii) only apply, when there exists no path between the two vertices, 
adding this edge will not generate a cycle. Induction on the number of added edges yields then $st$-planarity for $G'$.

Consider now an $st$-ordering $\pi \in \Pi(G')$. Since clearly $E' \subseteq E \cup E'$ holds,
$\pi$ is also an $st$-ordering for $G$, that is, $\Pi(G') \subseteq \Pi(G)$ holds. 
Recall that we constructed $G'$ such that 
for every $u \in V$ with $S(u) = \{v_1, \ldots, v_m\}$, 
there exists $v_1 \rightsquigarrow v_2 \rightsquigarrow \cdots \rightsquigarrow v_h \in G'$
and $v_m \rightsquigarrow v_{m-1} \rightsquigarrow \cdots \rightsquigarrow v_h  \in G'$.
It follows that for every $\pi \in \Pi(G')$
\[
	\gforall 1 \leq i < h : \pi(v_{i}) < \pi(v_{i+1}) \; \wedge \; \gforall h \leq i < m :  \pi(v_{i}) > \pi(v_{i+1})
\] holds, which implies that $S(u)$ is bitonic with respect to $\pi$.
Since this holds for all $u\in V$, it follows that $\Pi(G') \subseteq \Pi_b(G)$.
Moreover, $G'$ has at least one $st$-ordering, that is, $\Pi(G') \neq \emptyset$, thus, there exists
$\pi \in \Pi_b(G)$.
\qed\end{proof}

Let us summarize the implications of the lemma. The only requirement is that the
graph complies with our necessary condition, that is, the absence of forbidden configurations. 
If this is the case, then Lemma~\ref{lem:bitonic_planar_st_graph_ordering}
provides us with a bitonic $st$-ordering, which in turn proves that this condition is sufficient.
\begin{equation*}
\begin{split}
\exists \pi \in \Pi_b(G) \Leftrightarrow&\; \forall u \in V \text{ with } S(u)=\{v_1, \ldots, v_m \}\\
 &\;\gforall 1\leq i < j < m \;:\;  v_{i+1} \npathto v_{i} \vee v_{j} \npathto v_{j+1} \\
\end{split}
\end{equation*}

\begin{algorithm}[t]
\SetKwInOut{Input}{input}
\SetKwInOut{Output}{output}
\SetKwData{decreasing}{decreasing}
\SetKw{KwTrue}{true}
\SetKw{KwFalse}{false}
\SetKw{UpTo}{to}
\Input{Embedded planar $st$-graph $G = (V,E)$ with $S(u)$ for every $u \in V$.}
\Output{If exists, a bitonic $st$-ordering $\pi$ for $G$.}
\Begin{

	$E' \gets \emptyset$\;
	\For{$u \in V$ with $S(u) = \{v_1, \ldots, v_m\}$ }
	{
		$\textit{decreasing} \gets \KwFalse$\;
		\For{$i = 1$ \UpTo $m-1$}{
			$w \gets \textsc{faceSink}(u, v_i, v_{i+1})$\;
%			\lIf{$w = v_{i+1} \text{ \bf and } \textit{decreasing}$}{\Return \textsc{reject}}
%			\lIf{$w = v_{i}$}{$\textit{decreasing}\gets \KwTrue$}
			$\textbf{if }w = v_{i+1} \textbf{ and } \textit{decreasing} \text{ \bf then }$\Return \textsc{reject}\;
			$\textbf{if }w = v_{i} \textbf{ then }\textit{decreasing}\gets \KwTrue$\;
			\If{$v_{i} \neq w \neq v_{i+1}$}{
				$\textbf{if }\textit{decreasing}\textbf{ then } E' \gets E' \cup (v_{i+1}, v_{i}) \textbf{ else } E' \gets E' \cup (v_{i}, v_{i+1})$\;
			}
		}
	}
	compute $\pi \in \Pi(V, E \cup E')$\;
	\Return $\pi$
}
\caption{\small Recognition and ordering algorithm for planar-st graphs}\label{alg:bitonic_recognition}
\end{algorithm}

With a full characterization now at our disposal and in combination with Lemma~\ref{lem:bitonic_face_paths}, 
we are able to describe a simple linear-time algorithm 
(Algorithm~\ref{alg:bitonic_recognition}) which tests a given graph and in case it admits
a bitonic $st$-ordering, computes one. We iterate over $S(u)$ and as long as there is no path $v_{i+1} \pathto v_i$,
we assume $i < h$ and fill possible gaps. Once we encounter a path $v_{i+1} \pathto v_i$ for the first time,
we implicitly set $h = i$ via the flag and continue to add edges, but now the reverse ones.
But in case we find a path $v_{i} \pathto v_{i+1}$, then it forms with $v_{h+1} \pathto v_h$ a forbidden configuration 
and the graph can be rejected. If we succeed in all successor list, an $st$-ordering for $G'$ is computed,
which is a bitonic one for $G$. Since $G'$ is $st$-planar and has the same vertex set as $G$,
we can claim that the overall runtime is linear. Let us state this as the main result of this section.

\begin{theorem}\label{thm:bitonic_st_ordering_main_1}
Deciding whether an embedded planar $st$-graph $G$ admits a bitonic $st$-ordering $\pi$ or not is linear-time solvable.
Moreover, if $G$ admits such an ordering, $\pi$ can be found in linear time.
\end{theorem}

Next we will consider the case in which no bitonic $st$-ordering exists.
Although our initial motivation was to create upward planar straight-line drawings,
we now allow bends and shift our efforts to upward planar poly-line drawings.

%===================================================================================
\section{Upward planar poly-line drawings with few bends}
%===================================================================================
We start with a simple observation. Consider a forbidden configuration
consisting of two paths $v_{i+1} \rightsquigarrow v_{i}$ and $v_{j} \rightsquigarrow v_{j+1}$ 
with $i < j$ between successors of a vertex $u$ as shown in Fig.~\ref{fig:bitonic_forbidden_config}.
Notice that $(u,v_i)$ and $(u,v_{j+1})$ are transitive edges. Since a reduced planar $st$-graph has
no transitive edges, we can argue the following.
\begin{corollary}
Every reduced planar $st$-graph admits a bitonic $st$-ordering.
\end{corollary}

This leads to the idea to use the same transformation as Di~Battista~et~al.~\cite{Battista:1992fk} in their dominance-based 
approach. We can split every transitive edge to obtain a reduced planar $st$-graph and draw it upward planar straight-line.
Replacing the dummy vertices with bends results in an upward planar poly-line drawing with at most 
$2|V|-5$ bends, at most one bend per edge and quadratic area.

But we can do better using the following idea: If we have a single forbidden configuration, it suffices to split only one
of the two transitive edges. More specifically, if we split in Fig.~\ref{fig:bitonic_forbidden_config} 
the edge $(u,v_i)$ into two new edges $(u,v'_i)$ and $(v'_i, v_i)$ with $v'_i$ being the dummy vertex,
then $v'_i$ replaces $v_i$ in $S(u)$. But now there exists no path from $v_{i+1}$ to  $v'_i$, hence, the forbidden configuration
has been destroyed at the cost of one split. Moreover, a pair of transitive edges does not necessarily induce a forbidden configuration.
At this point the question arises how such a split affects other successor lists and if it may even create new forbidden configurations.
The following trivial observation is helpful in this regard.
%=================================================================
\newcommand{\lempathsplit}{
Let $G' = (V',E')$ be the graph obtained from splitting an edge $(u,v)$ of a graph $G = (V,E)$ by inserting a dummy vertex $v'$.
More specifically, let $V' = V \cup \{ v' \}$ and $E' = (E - ( u,v )) \cup \{ (u,v'),(v',v) \}$. Then
for all $w, x \in V$ there exists a path $w \rightsquigarrow x \in G$, if and only if there exists a path $w \rightsquigarrow x \in G'$.
}
%=================================================================
\begin{lemma}\label{lem:path_after_split}
\lempathsplit
\end{lemma}

Since a forbidden configuration is solely defined by the existence of paths, we can argue now
with Lemma~\ref{lem:path_after_split} that a split does not create nor resolves forbidden configurations
in other successor lists. However, one vertex that is not covered by the lemma is the dummy vertex itself, 
but it has only one successor which is insufficient for a forbidden configuration. This locality 
is of great value, because it enables us to focus on one successor list, instead of having to deal with a bigger picture.
Next we prove an upper bound on the number of edges to split in order to resolve all forbidden configurations.

\begin{lemma}\label{lem:bitonic_max_bends}
Every embedded planar $st$-graph $G = (V,E)$ can be transformed into a new one
that admits a bitonic $st$-ordering
by splitting at most \mbox{$|V|-3$ edges}.
\end{lemma}
\begin{proof}
Consider a vertex $u$ and its successor list $S(u) = \{ v_1, \ldots, v_m \}$ 
that contains multiple forbidden configurations of paths. Instead of arguing by means 
of forbidden configurations, we use our second condition from Proposition~\ref{pro:bitonic_h}, 
that is, the existence of a vertex $v_h$ such that every path that exists between 
two consecutive successors $v_i$ and $v_{i+1}$, 
is directed from $v_i$ towards $v_{i+1}$ for $i < h$, or 
from $v_{i+1}$ towards $v_{i}$ if $i \leq h$ holds.
Of course $h$ does not exist due to the forbidden configurations. 
But we can enforce its existence by splitting some edges.

Assume that we want $v_h$ to be the first successor, that is, $h = 1$. 
Then every path from $v_i$ to $v_{i+1}$ with $1 \leq i < m$ is in conflict with this choice.
We can resolve this by splitting every edge $(u,v_{i+1})$ for which a path $v_i \pathto v_{i+1}$ exists.
Clearly, the maximum number of edges to split is at most $m-1$, that is the
case in which for every $1 \leq i < m$, there exists a path from $v_i$ to $v_{i+1}$. 
However, there do not exist 
paths $v_i \pathto v_{i+1}$ and $v_{i+1} \pathto v_{i}$ at the same time, because $G$ is acyclic.
So, if the number of edges to split is more than $\frac{m-1}{2}$, then there are less than
$\frac{m-1}{2}$ paths of the form $v_{i+1} \rightsquigarrow v_{i}$. In that case, 
we may choose in a symmetric manner $v_h$ to be the last successor ($h = m$),
instead of being the first. Or in other words, we choose $v_h$ to be the first or the last successor, 
depending on the direction of the majority of paths. 
And as a result, at most $\frac{m-1}{2}$ edges have to be split.
Notice that the overall length of all successor lists is exactly the number of edges in the graph.
Hence, with $m = |S(u)|$ we get $\sum_{u \in V}{|S(u)|} = |E| \leq 3|V| - 6$, and the claimed upper bound can be derived by
\[
\sum_{u \in V}{\frac{|S(u)|-1}{2}} \leq \frac{3|V| - 6 - |V|}{2} = |V| - 3.
\] Moreover, the split procedure preserves $st$-planarity of $G$.
\qed\end{proof}

\begin{figure}[t]
\centering
    \begin{minipage}[b]{.33\textwidth}
        \centering
        \subfloat[\label{fig:bitonic_sharpness_0}{}]
        {\includegraphics[page=1, scale=0.85]{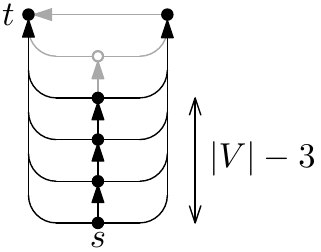}}
    \end{minipage}
    \begin{minipage}[b]{.66\textwidth}
        \centering
        \subfloat[\label{fig:bitonic_min_edges}{}]
        {\includegraphics[page=1, scale=0.9]{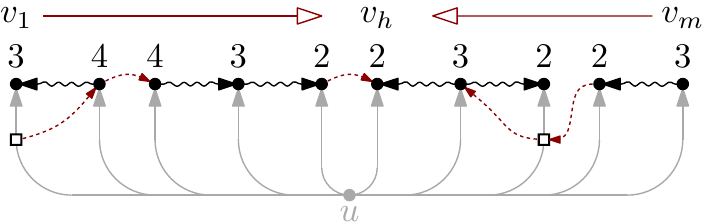}}
    \end{minipage}
    \caption{(a)~Example of a graph with $|V|-3$ forbidden configurations, each requiring one split to be resolved.
                  (b)~Example for finding the smallest set of edges to split. The numbers indicate how many splits are 
                  necessary when choosing the corresponding vertex to be $v_h$.
                  For $v_5, v_6, v_8$ and $v_9$ only two splits are necessary. 
                  Choosing $h = 6$ results in $E_{\text{split}} = \{ (u,v_1), (u,v_8)\}$. 
                  The squares indicate the result of the two splits, whereas the dotted edges represent $E'$ in Algorithm~\ref{alg:bitonic_recognition}.}\label{fig:bitonic_sharpness_1}
\end{figure}
%For $v_5, v_6, v_8$ and $v_9$ only two splits are necessary.

One may wonder now if this bound can be improved. Unfortunately, the
graph shown in Fig.~\ref{fig:bitonic_sharpness_0} is an example that requires $|V| - 3$ splits, hence, the bound is tight.
It also shows that there exist graphs that can be drawn upward planar straight-line in polynomial area but do not admit a bitonic $st$-ordering.
But we will push the idea of splitting edges a bit further from a practical point of view, and
focus on the problem of finding a minimum set of edges to split. 

In the following we describe an algorithm that solves this problem in linear time.
To do so, we introduce some more notation.
Let $u \in V$ be a vertex with successor list $S(u) = \{v_1, \ldots, v_m\}$.
We define 
$L(u,h) = |\{ i < h : v_{i+1} \rightsquigarrow v_{i} \}|$ and $R(u,h) = |\{ i < h : v_{i} \rightsquigarrow v_{i+1} \}|.$
If we choose now a particular $1 \leq h \leq m$ at $u$, then we have to split
every edge $(u,v_{i+1})$ with $i < h$ for which there exists a path $v_{i+1} \rightsquigarrow v_{i}$, and 
every edge $(u,v_{i})$ with $h \leq i$ for which $G$ contains a path $v_{i} \rightsquigarrow v_{i+1}$,
that is, we have to split $L(u,h) + R(u,m) - R(u,h)$ edges. See Fig.~\ref{fig:bitonic_min_edges} for an example.
When now considering all successor lists,
the minimum number of edge splits is
\[
\sum_{u \in V} \left(R(u,m) + \min_{1 \leq h \leq m }\left\{ L(u,h) - R(u,h)\right\} \right).
\]
Notice that the locality of a split allows us to minimize the number of edge splits for every successor list independently.
From an algorithmic point of view, we are interested in the value of $h$ and not in the number of splits, hence, we may drop $R(u,m)$
and consider the problem of finding $h$ for which $L(u,h) - R(u,h)$ is minimum. 
Since this is now only a matter of counting paths for which we can again
exploit Lemma~\ref{lem:bitonic_face_paths}, a linear-time algorithm becomes straightforward (see Algorithm~\ref{alg:bitonic_min_split}).
And as a result, we may state the following lemma without proof.

\begin{algorithm}[t!]
\SetKwInOut{Input}{input}
\SetKwInOut{Output}{output}
\SetKwData{decreasing}{decreasing}
\SetKw{KwTrue}{true}
\SetKw{KwFalse}{false}
\SetKw{UpTo}{to}
\Input{Embedded planar $st$-graph $G = (V,E)$ with $S(u)$ for every $u \in V$.}
\Output{Minimum set $E_{\text{split}} \subset E$ to split for admitting a bitonic $st$-ordering.}
\Begin{
	$E_{\text{split}} \gets \emptyset$\;
	\For{$u \in V$ with $S(u) = \{v_1, \ldots, v_m\}$ }
	{
		$h \gets 1$\;
		$c_{\min} \gets c \gets 0$\;
		
		\For{$i = 2$ \UpTo $m$}{		
			$w \gets \textsc{faceSink}(u, v_{i-1}, v_{i})$\;		
			%\lIf{$w = v_{i-1}$}{$c \gets c + 1$}
			%\lIf{$w = v_{i}$}{$c \gets c - 1$}
			$\textbf{if }w = v_{i-1} \textbf{ then } c \gets c + 1$\;
			$\textbf{if }w = v_{i} \textbf{ then } c \gets c - 1$\;
			\If{$c < c_{\min}$}{ 
				$c_{\min} \gets c$\;
				$h \gets i$\;
			}
				
		}

		\For{$i = 1$ \UpTo $h-1$}{
			%\lIf{$v_{i} =  \textsc{faceSink}(u, v_i, v_{i+1})$}{$E_{\text{split}} \gets E_{\text{split}} \cup (u, v_{i}) $}
			$\textbf{if }v_{i} =  \textsc{faceSink}(u, v_i, v_{i+1}) \textbf{ then } E_{\text{split}} \gets E_{\text{split}} \cup (u, v_{i}) $\;
		}		
		
		\For{$i = h$ \UpTo $m-1$}{
			%\lIf{$v_{i+1} =  \textsc{faceSink}(u, v_i, v_{i+1})$}{$E_{\text{split}} \gets E_{\text{split}} \cup (u, v_{i+1}) $}
			$\textbf{if }v_{i+1} =  \textsc{faceSink}(u, v_i, v_{i+1}) \textbf{ then } E_{\text{split}} \gets E_{\text{split}} \cup (u, v_{i+1})  $\;
		}		
	}
	\Return $E_{\text{split}}$
}
\caption{\small Algorithm for computing the minimum set of edges to split.}\label{alg:bitonic_min_split}
\end{algorithm}

\begin{lemma}\label{lem:bitonic_min_split}
Every embedded planar $st$-graph $G = (V,E)$ can be transformed into a planar $st$-graph that admits
a bitonic $st$-ordering by splitting every edge at most once.
Moreover, the minimum number of edges to split is at most $|V| - 3$ and they can be found
in linear time.
\end{lemma}

\noindent Now we may use this to create upward planar poly-line drawings with few bends.

\begin{theorem}\label{thm:bitonic_min_split}
Every embedded planar $st$-graph $G = (V,E)$ admits an upward planar poly-line drawing
within quadratic area having at most one bend per edge, at most $|V|-3$ bends in total, and
such a drawing can be obtained in linear time.
\end{theorem}
\begin{proof}
We use Lemma~\ref{lem:bitonic_min_split} to obtain a new planar $st$-graph $G'=(V',E')$ with 
$|V'| \leq  2|V| -3$ and a corresponding bitonic $st$-ordering $\pi$ with Algorithm~\ref{alg:bitonic_recognition}.
With Theorem~\ref{thm:bitonic_upward_straightline},
an upward planar straight-line layout of size $(2|V'|-2) \times (|V'|-1)$ for $G'$ is computed. Replacement of the dummy vertices by bends,
yields an upward planar poly-line drawing for $G$ of size at most $(4|V|-8) \times (2|V|-4)$.
\qed\end{proof}

\noindent Recall that every upward planar graph is a spanning subgraph of a planar $st$-graph~\cite{DiBattista1988175}.
Therefore, the bound of $|V|-3$ translates to all upward planar graphs.

\begin{corollary}
Every upward planar graph $G=(V,E)$ admits an upward planar poly-line drawing
within quadratic area having at most one bend per edge and at most $|V|-3$ bends in total.
\end{corollary}

%===================================================================================
\section{Conclusion}\label{se:conclusions}
%===================================================================================
In this work we have introduced the bitonic $st$-ordering for planar $st$-graphs.
Although this technique has its limitations, it provides the properties of canonical orderings for the directed case.
We have shown that this concept is viable by using a classic undirected incremental drawing algorithm for creating
upward planar drawings with few bends.

\bibliographystyle{splncs03}
\bibliography{references}

\opt{arxiv}{% !TEX root = bitonic_upward.tex
% Please leave the above for my texshop
\newpage
\appendix
\section*{Appendix}

\section{Omitted proofs}\label{app:proofs}
\rephrase{Lemma}{\ref{lem:lembitonicdrawing}}{\lembitonicdrawing}
\begin{proof}
The first and second statement hold for every $st$-ordering with $s$ and $t$ on the outer face.
For the third statement assume to the contrary, that for some $1 < k \leq |V|$ 
the neighbors of a vertex $v$ with $\pi(v) \leq k$
that are in $G - G_k$ do not appear consecutively in the embedding around $v$. 
Then $v$ has two successors $w_a, w_c \in S(v)$
with $\pi(w_a) > k$ and $\pi(w_c) > k$. Assume that $w_a$ precedes $w_c$ in $S(v)$, that is $a < c$.
Since all vertices in $S(v)$ appear consecutively in the embedding, there exists 
then a third successor $w_b$ between $w_a$ and $w_c$ in $S(v)$ that by our assumption is in $G_{k}$, that is, $\pi(w_b) \leq k$ holds. 
Notice that $S(v)$ is of the form $S(v) = \{ \ldots, w_a, \ldots, w_b, \ldots, w_c, \ldots \}$ and 
$\pi(w_a) > \pi(w_b) < \pi(w_c)$ holds, which contradicts that $S(v)$ is bitonic with respect to $\pi$.
\qed\end{proof}

\rephrase{Lemma}{\ref{lem:bitonic_alternative}}{\lembitonic}
\begin{proof}
Recall that $A$ is bitonic increasing if and only if 
there exists $1 \leq h \leq n$ such that $a_1 < \cdots < a_h > \cdots > a_n$ holds.
We first prove~\lq\lq$\Rightarrow$\rq\rq, that is, if $A$ is bitonic increasing,
then there exists no pair $i,j$ with $1 \leq i < j < n$ and $a_i > a_{i+1} \wedge a_j < a_{j+1}$.
Assume to the contrary that there exists such a pair. Then from $a_i > a_{i+1}$, it follows
that $h \geq i$, and $a_j < a_{j+1}$ yields $j < h$, which contradicts $i < j$.
For~\lq\lq$\Leftarrow$\rq\rq~we choose, if it exists, $h = \min\{j \; | \; a_j > a_{j+1}\}$, otherwise
we set $h = n$. By our choice of $h$, $a_i < a_{i+1}$ holds for every $1 \leq i < h$. 
Moreover, for every $h \leq j < n$, it must hold that $a_j > a_{j+1}$,
because otherwise, there exists $1 \leq h < j < n$ with $a_h > a_{h+1} \wedge a_j < a_{j+1}$.
\qed\end{proof}

\rephrase{Proposition}{\ref{pro:bitonic_h}}{\probitonicpaths}
\begin{proof}
We argue the same way as in the proof of Lemma~\ref{lem:bitonic_alternative}.
If there exists no path $v_{i+1} \rightsquigarrow v_{i}$ with $1 \leq i < m$, choose $h = m$.
Then $\gforall 1 \leq i < m : v_{i+1} \npathto v_{i}$ is satisfied in a trivial way.
If there exists at least one such path, we set ${ h = \min\{ i \mid  v_{i+1} \pathto v_{i} \} }$ which
satisfies $\gforall 1 \leq i < h : v_{i+1} \npathto v_{i}$ by construction.
Now assume to the contrary that there exists a path $v_{j} \pathto v_{j+1}$ with $h \leq j < m$.
Then there exists $v_{h+1} \rightsquigarrow v_{h}$ and $h \leq j$ holds, which contradicts our assumption that
for every $1\leq i < j < m$, it holds that $ v_{i+1} \npathto v_{i} \vee v_{j} \npathto v_{j+1}$.
\qed\end{proof}

\rephrase{Lemma}{\ref{lem:path_after_split}}{\lempathsplit}
\begin{proof}
Notice that $w,x \in V$ implies $w \neq v'$ and $x \neq v'$. Every path in $G$ that contains $(u,v)$ can use $(u,v'),(v',v)$ in $G'$.
Assume there is a path $w \rightsquigarrow x$ in $G'$ that does not exist in $G$, thus, it contains $(u,v')$ or $(v',v)$.
From $w \neq v' \neq x$, it follows that the path contains both edges, $(u,v')$ and $(v',v)$, and that they appear consecutively.
Hence, $w \rightsquigarrow x$ can use the edge $(u,v)$ in $G$ instead.
\qed\end{proof}

\section{Description of the upward planar straight-line algorithm}\label{app:desc}
In the following, we describe how to adapt the canonical ordering-based planar straight-line algorithm to bitonic $st$-orderings
by borrowing some ideas from Harel and Sardas~\cite{harel98algorithm}. 
They first describe a linear-time algorithm to compute a 
biconnected canonical ordering.
Then a modification of the algorithm of de~Fraysseix et al.
is used to obtain a planar straight-line layout.
The key observation is that when installing a vertex $v_k$ that has at least two neighbors on the contour $C_{k-1}$,
one can proceed as in the original algorithm.
The only problematic case is the one in which
a vertex $v_k$ has only one neighbor on $C_{k-1}$, say $w_i$. 
Harel and Sardas~\cite{harel98algorithm} introduce  
the property of having \emph{left, right} and \emph{legal support} for these vertices.
Their solution to the problem is as follows:
If $v_k$ has left support at its only neighbor $w_i$, 
then one may use $w_{i-1}$, the predecessor of $w_i$ on $C_{k-1}$,
as a second neighbor for $v_k$ and proceed as in the original 
algorithm by pretending that the edge $(v_k, w_{i-1})$ exists. 
However, this is only possible, because the property of having left support guarantees that 
all edges that have to be attached to $w_i$ later, follow $(v_k, w_i)$ clockwise in the embedding.
Roughly speaking, all edges to be attached later appear to the right of $v_k$, so $v_k$ is 
placed to the left of $w_i$ to keep $w_i$ accessible from above.
Similarly, when $v_k$ has right support, every edge incident to $w_i$ that is not yet present 
will be attached from the left. Therefore, in case of right support, we may use $w_{i+1}$ 
as a second neighbor for $v_k$. An example for having right support is given in Fig.~\ref{fig:fpp_bitonic_support_1}.

It is not difficult to see that due to the third statement in Lemma~\ref{lem:lembitonicdrawing}, 
we can use the idea of Harel and Sardas 
to deal with the case in which a vertex has only a single predecessor.
When placing such a vertex, say $v_k$, whose only predecessor is $u$, then 
we can assume that $v_k$ is not preceded and followed in $S(u)$ by vertices with a label greater than $k$.
Therefore, the concept of having left and right support translates to bitonic $st$-orderings in the following sense: 
$v_k$ has left support (at $u$) if no vertex preceding $v_k$ in $S(u)$ exists with a label greater than $k$.
And in a symmetric manner, $v_k$ has right support, if there is no vertex following $v_k$ in $S(u)$ with a label greater than $k$.

However, one problem arises: The approach by Harel and Sardas requires a vertex with only one neighbor on $C_{k-1}$ to have legal support,
not just left or right support. A quick look at their definition reveals that there is only a difference at the boundary of the contour.
More specifically, if the only predecessor of $v_k$ is $w_1$ (or $w_m$), then $v_k$ must have right support (or left support, respectively).
This is not necessarily the case in a bitonic $st$-ordering, 
where it may happen for example that $v_k$ has right support at $w_m$.
Let us assume for a moment that we have to cope with this case in which $v_k$ has right support at $w_m$.
Hence, the edge $(v_k, w_m)$ must have a slope of $+1$, 
thus, we are forced to choose $w_l = w_m$, whereas for $w_r$ 
we are then not able to find an appropriate vertex on $C_{k-1}$.
See Fig.~\ref{fig:fpp_steps_1} for an illustration of the problem of lacking legal support.

\begin{figure}[t]
  \centering
    \begin{minipage}[b]{0.49\textwidth}
        \centering
        \subfloat[\label{fig:fpp_steps_1}{}]
        { \includegraphics[page=1]{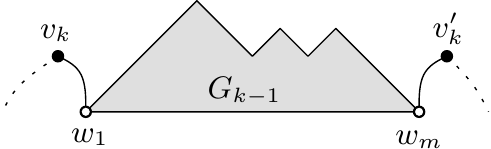}}
    \end{minipage}\hfill
    \begin{minipage}[b]{0.49\textwidth}
        \centering
        \subfloat[\label{fig:fpp_steps_2}{}]
        {\includegraphics[page=2]{fpp_steps}}
    \end{minipage}
\caption{(a) The problem of having no legal support at the boundary of the contour $C_{k-1} = \{ w_1, \ldots, w_m \}$.  The vertex to place has left support at $w_1$ or right support at $w_m$. (b) Two artificial vertices $v_L, v_R$, one at the beginning and one at the end of $C_{k-1} = \{ v_L = w_1, \ldots, w_m = v_R \}$ may serve as a second neighbor of $v_k$ in $G_{k-1}$.}
\end{figure}

To overcome this problem and without limiting the applicability of our bitonic $st$-ordering, 
we make a small modification to the algorithm.
We add two dummy vertices $v_L$ and $v_R$ that take the roles of $v_1$ and $v_2$ 
in the original algorithm with the property that $v_L$ is always 
the first, and $v_R$ always the last vertex in every contour, 
that is, for every $1 \leq k \leq n$, $C_k = \{ v_L = w_1, \ldots, w_m = v_R\}$ holds.
Notice that $v_L$ and $v_R$ are isolated vertices, 
thus, there exists no $v_k$ whose only predecessor is $v_L$ or $v_R$,
and that has left or right support.
Hence, we are always able to find a second neighbor on $C_{k-1}$ for $v_k$ as depicted in Fig.~\ref{fig:fpp_steps_2}. 

\begin{algorithm}[p]
\small
\SetKwInOut{Input}{input}
\SetKwInOut{Output}{output}
\SetKwInOut{Variables}{variables}
\SetKwData{decreasing}{decreasing}
\SetKw{KwTrue}{true}
\SetKw{KwFalse}{false}
\SetKw{UpTo}{to}
\SetKw{DownTo}{down to}
\Input{Embedded planar $st$-graph $G=(V,E)$ with successor lists $S(u)$ for every $u \in V$ and bitonic $st$-ordering $\pi$ for $G$.}
\Output{Grid-coordinates for an upward planar straight-line drawing.}
%\Variables{$v_k = v \in V$ with $\pi(v) = k$,
%Pair of dummy vertices $v_L, v_R$, the contour $C_k$ after installing $v_k$.}
\Begin{	
	$x(v_L)  \gets  0$; $y(v_L) \gets 0$\;
	$x(v_1)  \gets  1$; $y(v_1) \gets 1$\;
	$x(v_R) \gets  2$; $y(v_R) \gets 0$\;
	$C_1 \gets  \{v_L, v_1, v_R\}$\;
	\smallskip
	\tcp{bottom-up pass}
	\For{$k = 2$ \UpTo $ n$ }
	{
		$l \gets \min\{i~|~(w_i, v_k) \in E\}$\;
		$r \gets \max\{i~|~(w_i, v_k) \in E\}$\;
		\smallskip
		\tcp{one predecessor case}
		\If{l = r}{
			$v_p \gets$ preceding vertex of $v_k$ in $S(w_r)$\;
			%\lIf{$v_p = nil$ \text{\bf or} $\pi(v_p) \leq k$}{$l \gets l - 1$}
			$\textbf{if }v_p = nil \textbf{ or } \pi(v_p) \leq k\textbf{ then }l \gets l - 1$\;
			$v_s \gets$ following vertex of $v_k$ in $S(w_r)$\;
			%\lIf{$v_s = nil$ \text{\bf or} $\pi(v_s) \leq k$}{$r \gets r + 1$}
			$\textbf{if }v_s = nil \textbf{ or } \pi(v_s) \leq k\textbf{ then }r \gets r + 1$\;
		}
		\smallskip
		\tcp{distance $w_l \leftrightarrow w_r$ after shift}
		$d \gets 2 + \sum_{i = l+1}^r x(w_i)$\;
		\smallskip
		\tcp{place $v_k$}
		$x(v_k) \gets (d + y(w_r) - y(w_l))/{2}$\;
		$y(v_k) \gets (d + y(w_r) + y(w_l))/{2}$\;
		\smallskip
		\tcp{offset $w_{l+1}, \ldots, w_{r-1} \leftrightarrow v_k$}
		$t \gets 1  - x(v_k)$\;
		\For{$i = l+1$ \UpTo  $r-1$ }{
			$\textit{parent}(w_i) \gets v_k$\;
			$t \gets t + x(w_i)$\;
			$x(w_i) \gets t$\;
		}
		\smallskip
		$x(w_r) \gets d - x(v_k)$\;
		$C_{k} \gets$ replace $w_{l+1}, \ldots, w_{r-1}$ in $C_{k-1}$ with $v_k$ 
	}
	\smallskip
	\For{$i = 2$ \UpTo $|C_n|$ }
	{
		$x(w_i) \gets x(w_i) + x(w_{i-1})$
	}
	\smallskip
	\tcp{top-down pass}
	\For{$k = n$ \DownTo $1$ }
	{
		%\lIf{$\text{parent}(v_k) \neq nil$}{$x(v_k) = x(v_k) + x(parent(v_k))$}
		$\textbf{if }\text{parent}(v_k) \neq nil\textbf{ then }x(v_k) = x(v_k) + x(parent(v_k))$\;
	}
}
\caption{Shifting method for bitonic $st$-orderings}\label{alg:shifting_method}
\end{algorithm}

Now we put these ideas together by describing an algorithm (see Algorithm~\ref{alg:shifting_method}).
We start by placing $v_L, v_1$ and $v_R$ at $(0,0), (1,1)$ and $(2,0)$, respectively.
In every step $2 \leq k \leq n$, we proceed exactly as in the canonical ordering based variant 
only the subroutine for determining $w_l$ and $w_r$ has to be adjusted according to the idea
of Harel and Sardas. However, notice that if $v_k$ has left and right support at $w_i$, then $w_l = w_{i-1}$ and  $w_r = w_{i+1}$ is chosen.
A complete example is shown in Fig.~\ref{fig:bitonic_algo_example},
in which the drawing for a small graph with seven vertices is created step by step. The output of the algorithm for a larger example is given in Fig.~\ref{fig:bitonic_upward_polyline_output}.

\begin{figure}[t]
\centering
    \begin{minipage}[b]{0.2\textwidth}
        \centering
        \subfloat[\label{fig:fpp_bitonic_example_1}{}]
        { \includegraphics[page=1, scale = 0.95]{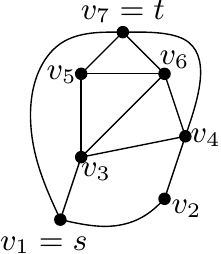}
        }
    \end{minipage}\hfill
    \begin{minipage}[b]{0.18\textwidth}
        \centering
        \subfloat[\label{fig:fpp_bitonic_example_2}{}]
        { \includegraphics[page=2]{fpp_bitonic_example}}
    \end{minipage}\hfill
    \begin{minipage}[b]{0.3\textwidth}
        \centering
        \subfloat[\label{fig:fpp_bitonic_example_3}{}]
 	{ \includegraphics[page=3]{fpp_bitonic_example}}
    \end{minipage}\hfill
    \begin{minipage}[b]{0.31\textwidth}
        \centering
        \subfloat[\label{fig:fpp_bitonic_example_4}{}]
 	{ \includegraphics[page=4]{fpp_bitonic_example}}
    \end{minipage} 
    \begin{minipage}[b]{0.45\textwidth}
        \centering
        \subfloat[\label{fig:fpp_bitonic_example_5}{}]
 	{ \includegraphics[page=5]{fpp_bitonic_example}}
    \end{minipage}\hfill
    \begin{minipage}[b]{0.5\textwidth}
        \centering
        \subfloat[\label{fig:fpp_bitonic_example_6}{}]
 	{ \includegraphics[page=6]{fpp_bitonic_example}}
    \end{minipage}\\
     \begin{minipage}[b]{0.468\textwidth}
        \centering
        \subfloat[\label{fig:fpp_bitonic_example_7}{}]
 	{ \includegraphics[scale=0.9, page=7]{fpp_bitonic_example}}
    \end{minipage}\hfill
    \begin{minipage}[b]{0.53\textwidth}
        \centering
        \subfloat[\label{fig:fpp_bitonic_example_8}{}]
 	{ \includegraphics[scale=0.9,page=8]{fpp_bitonic_example}}
    \end{minipage} 
    \caption{(a) Example graph consisting of seven vertices with a bitonic $st$-ordering.
    		  (b)-(h) Steps during the construction of the drawing. (b) $v_2$ is supported by $v_R$ and serves
		  in the next step (c) as supporting vertex for $v_3$. (f) $v_5$ uses $v_1$ as support.}\label{fig:bitonic_algo_example}
\end{figure}

\begin{figure}[p]
\centering
\includegraphics[angle=90,width=0.75\textwidth]{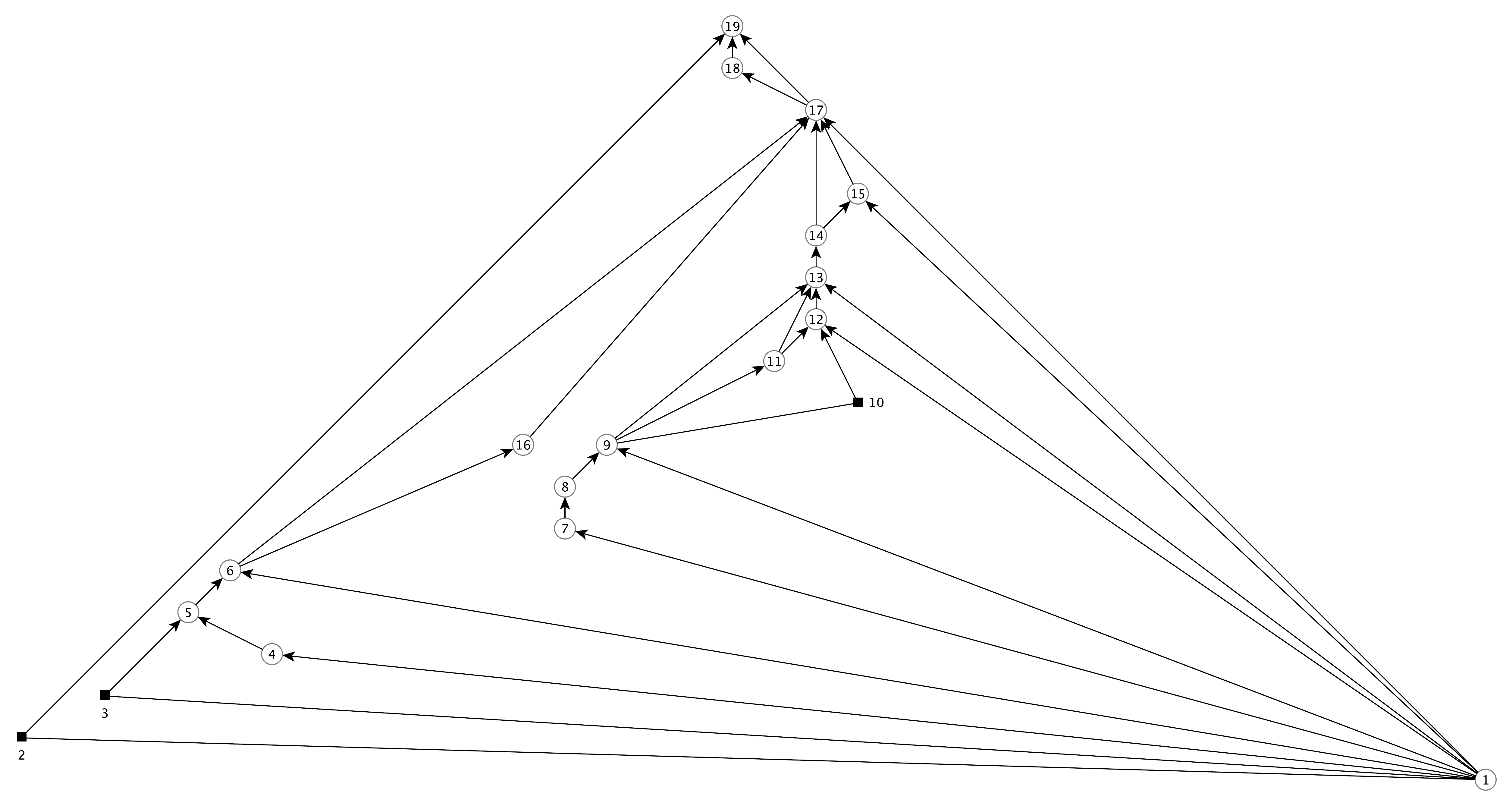}
\caption{Example of an upward planar poly-line drawing of a planar $st$-graph $G = (V,E)$ with $|V| = 16$ and $|E| = 30$.
Circles represent vertices of $G$, whereas squares indicate bends. The labels correspond to the rank in the bitonic $st$-ordering.}\label{fig:bitonic_upward_polyline_output}
\end{figure}

}
\end{document}